\documentclass[reqno,tbtags,intlimits,a4paper,oneside,12pt]{amsart}
\usepackage[cp1251]{inputenc}%
\usepackage[russian]{babel}
\usepackage{amssymb,upref,calc,url}
\usepackage[mathscr]{eucal}
\usepackage{graphicx}
\usepackage{amsmath, amsfonts,amscd,amsthm,amssymb,verbatim}
\usepackage[all]{xy}
\usepackage{textcomp}
\usepackage[in]{fullpage}
\newtheorem{thm}{Theorem}[section]
\newtheorem{cor}[thm]{Corollary}
\newtheorem{lem}[thm]{Lemma}

\theoremstyle{definition}
\newtheorem{defn}[thm]{Definition}

\newtheorem{ex}[thm]{Example}
\newtheorem{exs}[thm]{Examples}

\def\bcdot{\,\boldsymbol\cdot\,}

\usepackage{color}

\def\beq#1#2\eeq{%
        \begin{equation}%
        \label{#1}%
            #2%
        \end{equation}%
    }

\begin{document}
\title[The Mumford Dynamical System\\ and Hyperelliptic Kleinian Functions]
{The Mumford Dynamical System and\\ Hyperelliptic Kleinian Functions}
\author{Victor Buchstaber}
\address{V.A. Steklov Mathematical Institute\\ Gubkin St. 8\\
Moscow, 119991, Russia} \email{buchstab@mi-ras.ru}
%\thanks{Исследование выполнено за счет гранта Российского научного фонда (проект 14-50-00005).}

\maketitle

%\vspace{0.5cm}
\rightline{\em In memory of Israel Moiseevich Gelfand (1913--2009)}

\bigskip

%УДК 515.178.2+517.958+512.77

\begin{abstract}
We establish differential-algebraic theory of the Mumford dynamical system.
In the framework of this theory, we introduce the $(P,Q)$-recursion, which defines a sequence of functions $P_1,P_2,\ldots$
given the first function of this sequence $P_1$ and a sequence of parameters $h_1,h_2,\ldots$.
The general solution of the $(P,Q)$-recursion is shown to give a solution for the parametric graded Korteweg--de Vries hierarchy.
We prove that all solutions of the Mumford dynamical $g$-system are determined by the $(P,Q)$-recursion
under the condition $P_{g+1} = 0$, which is equivalent to an ordinary nonlinear differential equation of order $2g$ for the function $P_1$.
Reduction of the $g$-system of Mumford to the Buchstaber--Enolskii--Leykin dynamical system is described explicitly,
and its explicit $2g$-parameter solution in hyperelliptic Klein functions is presented.
\end{abstract}

%\tableofcontents

%%%%%%%%%%%%%%%%%
\section*{Introduction}
The theory and applications of the Korteweg--de Vries equation, see \cite{KdV-1895}, and the KdV hierarchy
were at the center of Israel Moiseevich Gelfand's attention for many year (\cite{G-D-75}--\cite{G-D-79}).

Solutions of the KdV hierarchy, obtained through quantum field theory methods and algebraic-geometric methods of soliton theory, are widely known.
Constructing its solutions and studying their properties are the focus of the analytical direction of this hierarchy's theory.
The methods include various recurrence constructions for the KdV hierarchy based on the connection of the Korteweg--de Vries equation
with fundamental problems in mathematical physics, classical and quantum mechanics, and the theory of commuting operators.
The differential-algebraic direction explores the construction of KdV hierarchies, their structure, and connections
with other fundamental equations and hierarchies within the theory of integrable systems. Deep problems in differential algebra,
differential equations theory, differential geometry, symplectic and Poisson manifolds, algebraic geometry,
and the theory of Abelian functions have emerged at the intersection of these directions, as seen in papers, reviews,
and monographs such as \cite{AdMoer-94}, \cite{B-Mikh-21}, \cite{Dic-03}--\cite{DKN-85}, \cite{GGKM-67}, \cite{G-D-79}--\cite{N-74}, and \cite{Tsig-05}.
Gelfand's interest in almost all aspects of the KdV theory and his results greatly contributed to the modern state of this theory.
In this paper, within the framework of the differential-algebraic direction of the KdV theory, we present the construction
of $(P, Q)$-recursion based on the Mumford dynamical system. Using the results of \cite{B-23}, we describe connections
of this recursion to the Gelfand--Dikii recursion (see \cite{G-D-75}, \cite{G-D-79}).

Based on the theta-functional solutions of the $g$-stationary Korteweg--de Vries hierarchy,
see \cite{N-74}, Dubrovin and Novikov showed \cite{D-N-74} that the universal Jacobian space of hyperelliptic curves
of genus $g$ (the \emph{DN space}) is birationally equivalent to $\mathbb{C}^{3g+1}$.

Using the theory of non-singular hyperelliptic curves and the theta-functional theory of their Jacobians,
Mumford introduced a dynamical system on $\mathbb{C}^{3g+1}$ and obtained an algebraic-geometric description
of hyperelliptic Jacobians in the solution space of this system, which is equivalent to the DN space,
see \cite[Chapter ~3, \S3, Theorem 3.1]{Mumf-84}.

In the framework of the algebraic-geometric theory of integrable systems, Vanhaecke and several other authors
further developed the theory of the Mumford dynamical system and its generalizations, see \cite{Vanh-01}, \cite{Fit-22},
as well as references cited in those works.

Buchstaber, Enolskii, and Leykin introduce and describe \cite{BEL-97-2} on the space $\mathbb{C}^{3g+1}$:

1) A polynomial dynamical system BEL of the form
\[
\mathcal{D}_\eta L_\xi = [L_\xi, M_{\xi,\eta}], \qquad \mathcal{D}_\eta = \sum_{i=1}^g \eta^{g-i}\partial_i,\; \partial_i = \frac{\partial}{\partial t_i},
\]
where $L_\xi$ and $M_{\xi,\eta}$ are matrix functions on $\mathbb{C}^g$, $g\geqslant 1$, that take values
in the Lie algebra $\mathcal{S}L(2,\mathbb{C}[\xi,\eta])$,
$\frac{\partial}{\partial t_i}\eta\equiv 0$, and $t_i$, $i=1,\dots,g$, are the coordinates in $\mathbb{C}^g$.

2) A form of the function $M_{\xi,\eta}$ for which this system can be integrated in hyperelliptic Klein functions.

3) The corresponding effective algebraic-geometric description of hyperelliptic Jacobians.

Section 1 of the present paper presents a differential-algebraic theory of the Mumford dynamical system,
see \eqref{f-2}--\eqref{f-4}, which does not rely on the theory of hyperelliptic curves and Abelian functions on their Jacobians.
For each function $P_1$ and a sequence of parameters $h_1, h_2, \dots$, we construct the so-called $(P, Q)$-recursion,
defining a sequence of functions $P_1, P_2, \ldots$. The solutions of the Mumford dynamical $g$-system are proven
to be determined by this recursion under the condition $P_{g+1} = 0$, which is equivalent to an ordinary nonlinear
differential equation of order $2g$ for the function $P_1$.
Theorem \ref{T-B1}, proved in \cite{B-23}, shows that the $(P, Q)$-recursion with a general parameter vector
provides a solution to the parametric KdV hierarchy. Theorem \ref{T-B2}, also proved in \cite{B-23},
demonstrates that, the special sequence of functions $P_1^0, P_2^0, \dots$ in the case of all parameters
being equal to $0$ can be scaled to a standard solution of the Gelfand--Dikii recursion.
We introduce the concept of a special solution to the KdV hierarchy, and show that this solution is determined
by a single parameter.
We prove the general solution of the $(P, Q)$-recursion (as well as the Gelfand--Dikii recursion)
to be a linear combination of special solutions with coefficients that are independent parameters.
An explicit expression for these parameters in terms of integrals of the Mumford dynamical system is obtained.

The key result of the paper is construction of differential equations that correspond to the $g$-stationary KdV hierarchy
and the $g$-Novikov equation in terms of the $(P, Q)$-recursion. As a result, we obtain solutions to these equations
for singular hyperelliptic curves.

In Section 8 of this paper, we describe a reduction of the Mumford dynamical system to the BEL dynamical system.
It is worth noting that the solution of the BEL dynamical system is based on multi-dimensional heat conduction equations
in a nonholonomic frame, as seen in \cite{BL-04}.
The effective construction of these solutions uses results obtained jointly with Bunkova, see \cite{B-Bun-20} and \cite{B-Bun-23}.

%%%%%%%%%%%%%%%   1
\section{The Mumford Dynamical System} \label{p-1}

%%%%%%%%%%%%%%%   1.1
\subsection{The Mumford Dynamical System in the Lax Form}\text{}

Consider the space $\mathbb{C}^{g}$, where $g\geqslant 1$, with coordinates $\mathbf{t}=(t_1,\ldots,t_g)$.
We introduce differentiation $\mathcal{D}_\eta = \sum_{i=1}^g \eta^{g-i}\partial_i$, where $\partial_i = \partial/\partial t_i$.
For independent parameters $\xi$ and $\eta$, $\mathcal{D}_\eta\xi = \mathcal{D}_\xi\eta \equiv 0$. We consider mappings
\[
L_\xi \colon \mathbb{C}^{g}\to \mathcal{S}L(2,\mathbb{C}[\xi])\quad \text{and}\quad A_\eta \colon \mathbb{C}^{g}\to \mathcal{S}L(2,\mathbb{C}[\eta]).
\]
We study dynamical systems of the form
\begin{equation}\label{f-1}
\mathcal{D}_\eta L_\xi(\mathbf{t}) = \frac{1}{\xi-\eta}[L_\xi(\mathbf{t}),L_\eta(\mathbf{t})] + [L_\xi(\mathbf{t}),A_\eta(\mathbf{t})],
\end{equation}
where $L_\xi(\mathbf{t}) =
\begin{pmatrix}
  v_\xi(\mathbf{t}) & u_\xi(\mathbf{t}) \\
  w_\xi(\mathbf{t}) & -v_\xi(\mathbf{t})
\end{pmatrix}.$

We will omit the argument $\mathbf{t}$ when the dependence on it is clear from the context.
Consider the space $\mathbb{C}^{3g+1}$, where $g\geqslant 1$, with coordinates
$(\mathbf{u},\mathbf{v},\mathbf{w})$, where $\mathbf{u}=(u_1,\ldots,u_g)$, $\mathbf{v}=(v_1,\ldots,v_g)$,
and $\mathbf{w}=(w_1,\ldots,w_{g+1})$. The mapping $L_\xi$ is defined by equations
\begin{equation}\label{f-1-1}
u_\xi = \xi^g+\sum_{i=1}^g u_i\xi^{g-i},\quad v_\xi = \sum_{i=1}^g v_i\xi^{g-i},\quad w_\xi = \xi^{g+1}+\sum_{i=1}^{g+1} w_i\xi^{g+1-i}.
\end{equation}

\begin{lem}\label{L-1}
Suppose that the mapping $L_\xi$ given by Equation \eqref{f-1-1} satisfies the system \eqref{f-1}.
Then $A_\eta(\mathbf{t}) = \begin{pmatrix}
  0 & 0 \\
  u_\eta & 0
\end{pmatrix}.$
\end{lem}

The system with such $A_\eta(\mathbf{t})$ we will call a Mumford system in Lax form.
It defines a dynamical system on $\mathbb{C}^{3g+1}$ that is quadratic in the variables $(\mathbf{u},\mathbf{v},\mathbf{w})$.

%%%%%%%%%%%%%%%   1.2
\subsection{The Dynamical System on $\mathbb{C}^{3g+1}$}\text{}

Let us write the Mumford system in an expanded form:
\begin{align}\label{f-2}
  \mathcal{D}_\eta u_\xi =& \frac{2}{\xi-\eta}(v_\xi u_\eta-u_\xi v_\eta); \\
\label{f-3} \mathcal{D}_\eta v_\xi =& \frac{1}{\xi-\eta}(u_\xi w_\eta - w_\xi u_\eta) + u_\xi u_\eta; \\
 \label{f-4} \mathcal{D}_\eta w_\xi =& \frac{2}{\xi-\eta}(w_\xi v_\eta - v_\xi w_\eta) - 2v_\xi u_\eta.
\end{align}

Directly from equations \eqref{f-2}--\eqref{f-4}, it follows
\begin{cor}\label{C-1}
In the conditions of Lemma \ref{L-1}, we obtain
\begin{itemize}
\item[1)] $\mathcal{D}_\eta u_\xi = \mathcal{D}_\xi u_\eta$;\vskip.3cm
\item[2)] $\mathcal{D}_\eta v_\xi = \mathcal{D}_\xi v_\eta$;\vskip.3cm
\item[3)] $\frac{\mathcal{D}_\xi w_\eta - \mathcal{D}_\eta w_\xi}{\xi-\eta} = \mathcal{D}_\eta u_\xi$.
\end{itemize}
\end{cor}

%%%%%%%%%%%%%%%   2
\section{Integrals of the Mumford Dynamical System} \label{p-2}
Set
\begin{equation*}\label{f-5}
H_\xi = -\det L_\xi = u_\xi w_\xi + v_\xi^2 = \xi^{2g+1} + \sum\limits_{n=1}^{2g+1}h_n \xi^{2g+1-n}.
\end{equation*}
Since the traces of the matrices in the right-hand side of Equation \eqref{f-1} are equal to zero,
then $\mathcal{D}_\eta H_\xi \equiv 0$ and
\begin{multline*}%\label{f-6}
H_\xi = \xi^{2g+1} + \sum\limits_{k=1}^{g}(u_k+w_k) \xi^{2g+1-k} + w_{g+1}\left( \xi^{g} + \sum\limits_{k=1}^{g}u_k\xi^{g-k} \right) + \\
      + \sum\limits_{m=2}\left( \sum\limits_{i+j=m}u_i w_j \right)\xi^{2g+1-m} + \sum\limits_{s=2}^{2g}\left( \sum\limits_{i+j=s}v_i v_j \right)\xi^{2g-s}.
\end{multline*}
Therefore,
\begin{align*}
\text{для } g\geqslant 1:\; & w_1 = h_1-u_1; \\
\text{для } g = 1:\; & w_2 = h_2-u_1w_1; \\
\text{для } g = 2:\; & w_2 = h_2-u_2-u_1w_1;\; w_3 = h_3-(u_1w_2+u_2w_1)-v_1^2; \\
\text{для } g > 2:\; & w_k = h_k-u_k-\sum\limits_{i+j=k}u_i w_j-\sum\limits_{i+j=k-1}v_i v_j,\; k=2,\ldots,g.
\end{align*}
Moreover, for $g > 2$,
\begin{equation*}
w_{g+1} = h_{g+1}-\sum\limits_{i+j=g+1}u_i w_j-\sum\limits_{i+j=g}v_i v_j.
\end{equation*}
\begin{cor}
The dynamical system \eqref{f-2}--\eqref{f-4} on $\mathbb{C}^{3g+1}$ with coordinates $(\mathbf{u},\mathbf{v},\mathbf{w})$
is equivalent to a family of compatible dynamical systems on $\mathbb{C}^{2g}$ with coordinates $(\mathbf{u},\mathbf{v})$ and parameters $\mathbf{h}^{g+1} = (h_1,\dots,h_{g+1})$.

For each fixed value of $\mathbf{h}^{g+1}\in\mathbb{C}^{g+1}$, the resulting dynamical system has $g$ algebraically independent polynomial integrals $h_{g+2},\ldots,h_{2g+1}$, where
\[
h_{g+k} = \sum\limits_{i+j=g+k}u_i w_j + \sum\limits_{i+j=g+k-1}v_i v_j,\; k=2,\ldots,g+1.
\]
\end{cor}

\begin{ex}
Let $g=1$. Then $u_\xi = \xi+u_1,\; v_\xi = v_1,\; w_\xi = \xi^2+w_1\xi+w_2$\, and\, $\mathcal{D}_\eta = \partial_1$.
Set $\partial_1 f = f',\; u=u_1$\; and\; $v=v_1$. We obtain a dynamical system on $\mathbb{C}^{4}$ described by equations
\begin{align*}
  u' =& \; -2v;\qquad \qquad \qquad w_1' = 2v;      \\
  v' =& \; u^2-uw_1+w_2; \qquad\, w_2' = -2(u-w_1)v
\end{align*}
with integrals $h_1=u+w_1,\; h_2=u_1w_1+w_2,\; h_3=uw_2+v^2$.

On $\mathbb{C}^{4}$ with coordinates $u$, $v$, $h_1$, $h_2$, we get a family of dynamical systems on $\mathbb{C}^{2}$ described by equations
\begin{equation}\label{f-12}
u'= -2v, \qquad v' = 3u^2-2h_1u_1+h_2
\end{equation}
with an integral $h_3 = u^3-h_1u^2+h_2u+v^2$.

Note that the polynomial $h_3 = h_3(u,v,h_1,h_2)$ is a Hamiltonian for the system \eqref{f-12}
with respect to the Poisson bracket $\{u,v\} = 1$, where $h_1$ and $h_2$ are Casimir functions.
\end{ex}

%\newpage
%%%%%%%%%%%%%%%%   3
\section{Poisson Brackets on $\mathbb{C}^{3g+1}$} \label{p-3}

%%%%%%%%%%%%%%%   3.1
\subsection{Special Poisson Brackets.}\text{}

This section is not formally related to the Mumford dynamical system but it is important for proving its integrability.

\begin{lem}\label{L-sP}
Let $\{\bcdot,\bcdot\}$ be a Poisson bracket defined in $\mathbb{C}^{3g+1}$ with coordinates $(\mathbf{u},\mathbf{v},\mathbf{w})$ such that
\[
\{u_\xi,u_\eta\} = \{v_\xi,v_\eta\} = 0,\quad \{u_\xi,v_\eta\} = \{u_\eta,v_\xi\},\quad \{u_\xi,w_\eta\} = \{u_\eta,w_\xi\}.
\]
Then, for the polynomial $H_\xi = v_\xi^2+u_\xi w_\xi$, the identity $\{H_\xi,H_\eta\} \equiv 0$ holds if and only if
\begin{equation*}%\label{f-sP}
u_\xi u_\eta \{w_\xi,w_\eta\} = 2(u_\xi v_\eta \{v_\eta,w_\xi\} - u_\eta v_\xi \{v_\xi,w_\eta\}).
\end{equation*}
\end{lem}
\begin{proof}
Let $\{\bcdot,\bcdot\}$ be a Poisson bracket in $\mathbb{C}^{3g+1}$ that satisfies the conditions of the lemma. Then,
\[
\{H_\xi,H_\eta\} = \{v_\xi^2,u_\eta w_\eta\} + \{u_\xi w_\xi,v_\eta^2\} + \{u_\xi w_\xi,u_\eta w_\eta\}.
\]
For the first term, we have
\[
\{v_\xi^2,u_\eta w_\eta\} = \{v_\xi^2,u_\eta\} w_\eta + u_\eta \{v_\xi^2,w_\eta\} = -2 v_\xi w_\eta \{u_\eta,v_\xi\} -2 u_\eta v_\xi \{w_\eta,v_\xi\}.
\]
Using the conditions on the Poisson bracket $\{\bcdot,\bcdot\}$, we obtain
\[
\{v_\xi^2,u_\eta w_\eta\} + \{u_\xi w_\xi,v_\eta^2\} = 2 u_\eta v_\xi \{v_\xi,w_\eta\} -2 u_\xi v_\eta \{v_\eta,w_\xi\}.
\]
For the third term,
\begin{multline*}
\{u_\xi w_\xi,u_\eta w_\eta\} = \{u_\xi w_\xi,u_\eta\} w_\eta + u_\eta \{u_\xi w_\xi,w_\eta\} = \\
= -u_\xi w_\eta \{u_\eta,w_\xi\} - u_\eta w_\xi \{w_\eta,u_\xi\} - u_\eta u_\xi \{w_\eta,w_\xi\}.
\end{multline*}
Using the conditions on the Poisson bracket $\{\bcdot,\bcdot\}$, we obtain
\[
\{u_\xi w_\xi,u_\eta w_\eta\} =  u_\xi u_\eta \{w_\xi,w_\eta\}.
\]
\end{proof}

%%%%%%%%%%%%%%%   3.2
\subsection{The Vanhaecke brackets.}\text{}

In monograph \cite{Vanh-01} Vanhaecke introduced a $(g+1)$-dimensional family of compatible Poisson brackets in $\mathbb{C}^{3g+1}$ that satisfy the conditions of Lemma \ref{L-sP}.
This family forms a $(g+1)$-dimensional space with a basis $\{\bcdot,\bcdot\}_k$, $k=0,1,\dots,g$, where
\begin{align*}
\{u_\xi,u_\eta\}_k =&\; \{v_\xi,v_\eta\}_k = 0, \\
\{u_\xi,v_\eta\}_k =&\; \frac{1}{\xi-\eta}(u_\xi\eta^k-u_\eta\xi^k), \\
\{u_\xi,w_\eta\}_k =&\; -\frac{2}{\xi-\eta}(v_\xi\eta^k-v_\eta\xi^k), \\
\{v_\xi,w_\eta\}_k =\; & \frac{1}{\xi-\eta}(w_\xi\eta^k-w_\eta\xi^k) - u_\xi\eta^k,\\
\{w_\xi,w_\eta\}_k =&\; 2(v_\xi\eta^k-v_\eta\xi^k).
\end{align*}

\begin{exs}
1) $\underline{g=1}:\; u_\xi = \xi+u,\; v_\eta = v; \quad \{u_\xi,v_\eta\}_k = \{u,v\}_k,\; k=0,1$:
\[
\{u,v\}_0 = 1,\quad \{u,v\}_1 = -u.
\]
2) $\underline{g=2}:\;  u_\xi = \xi^2+u_1\xi+u_2,\; v_\eta = v_1\eta+v_2$; \\
$\{u_\xi,v_\eta\}_k = \{u_1,v_1\}_k \xi\eta + \{u_1,v_2\}_k \xi + \{u_2,v_1\}_k\eta + \{u_2,v_2\}_k,\; k=0,1,2$,\\
$\{u_\xi,v_\eta\}_0 = \xi+\eta+u_1$:
\[
\{u_1,v_1\}_0 = 0,\quad \{u_1,v_2\}_0 = \{u_2,v_1\}_0 = 1,\quad \{u_2,v_2\}_0 = u_1.
\]
$\{u_\xi,v_\eta\}_1 = \xi\eta-u_2$:
\[
\{u_1,v_1\}_1 = 1,\quad \{u_1,v_2\}_1 = \{u_2,v_1\}_1 = 0,\quad \{u_2,v_2\}_1 = -u_2.
\]
$\{u_\xi,v_\eta\}_2 = -u_1\xi\eta-u_2(\xi+\eta)$:
\[
\{u_1,v_1\}_2 = -u_1,\quad \{u_1,v_2\}_2 = \{u_2,v_1\}_2 = -u_2,\quad \{u_2,v_2\}_2 = 0.
\]
\end{exs}

%%%%%%%%%%%%%%%%%%%%%%%  3.3
\subsection{Hamiltonians of the Mumford System and Casimirs of Brackets $\{\bcdot,\bcdot\}_k$.}\text{}

The coefficients $h_1,\dots,h_{2g+1}$ of the polynomial $H_\xi$ of degree $2g+1$ in $\xi$ are integrals of the Mumford dynamical system.
These coefficients are quadratic polynomials in $\mathbf{u},\mathbf{v},\mathbf{w}$; they are in involution with respect to the Poisson brackets $\{\bcdot,\bcdot\}_k$ for any $k=0,1,\dots,g$.

Set $\mathcal{D}_\eta(k) = -\eta^k \mathcal{D}_\eta$.
\begin{thm}
Using the bracket $\{\bcdot,\bcdot\}_k$, one can rewrite the Mumford dynamical system in Hamiltonian form for any $k=0,1,\ldots,g+1$ by setting
\begin{equation*}%\label{f-14-1}
\mathcal{D}_\eta(k) u_\xi = \{u_\xi,H_\eta\}_k, \quad \mathcal{D}_\eta(k) v_\xi = \{v_\xi,H_\eta\}_k, \quad \mathcal{D}_\eta(k) w_\xi = \{w_\xi,H_\eta\}_k.
\end{equation*}
\end{thm}
{\it Proof} follows directly from Equations \eqref{f-2}--\eqref{f-4} and formulas for the brackets $\{\bcdot,\bcdot\}_k$, $k=0,1,\dots,g+1$.

\begin{cor}
For the bracket  $\{\bcdot,\bcdot\}_k$, $k=0,1,\dots,g$,  the polynomials  $h_n = h_n(\mathbf{u},\break\mathbf{v},\mathbf{w})$ for $g-k+2\leqslant n\leqslant 2g-k+1$
are Hamiltonians of the Mumford dynamical system. The remaining polynomials $h_n$ from the list $h_1,\dots,h_{2g+1}$ are Casimirs of this bracket.
\end{cor}

\begin{exs}
1. $\underline{k=0}$. Hamiltonians for $g+2\leqslant n\leqslant 2g+1$; Casimirs for $1\leqslant n\leqslant g+1$.

2. $\underline{k=1}.$ Hamiltonians for $g+1\leqslant n\leqslant 2g$; Casimirs for $1\leqslant n\leqslant g$ and for $n=2g+1$.

3. $\underline{k=g}$. Hamiltonians for $2\leqslant n\leqslant g+1$; Casimirs for $n=1$ and for $g+2\leqslant n\leqslant 2g+1$.
\end{exs}

%%%%%%%%%%%%%%%%   4
\section{Differential Polynomials} \label{p-4}
Consider the space $\mathbb{C}^{3g+1}$ with coordinates $(\mathbf{u},\mathbf{v},\mathbf{w})$ as the phase space
of the flow defined by the differentiation $\partial_1$.
In this section, we first express the dynamical coordinates $v_k(t_1)$ and $w_k(t_1)$ in terms of the coordinates $u_k(t_1),\,u_k'(t_1),\,u_k''(t_1)$
and then introduce a recursion that expresses all dynamical coordinates in the form differential polynomials on $u_1(t_1)$.

Set $\partial_1 f = f'$.
\begin{thm}
Let $(u_\xi,v_\xi,w_\xi)$ be a solution of the system of Equations \eqref{f-2}--\eqref{f-4}. Then
\begin{align*}
 (a)\; v_\xi =&\; -\frac{1}{2}\mathcal{D}_\xi u_1 = -\frac{1}{2}u_\xi';\\
 (b)\; w_\xi =&\; (\xi+w_1-u_1)u_\xi + \mathcal{D}_\xi v_1 = (\xi+h_1-2u_1)u_\xi-\frac{1}{2}u_\xi''.
\end{align*}
\end{thm}
\begin{proof}
1). Dividing both sides of the equation \eqref{f-2} by $\xi^{g-1}$ and taking the limit at $\xi\to\infty$,
we obtain $\mathcal{D}_\eta u_1 = -2v_\eta$. Therefore, $\mathcal{D}_\xi u_1 = -2v_\xi$.
Thus, statement (a) follows from the corollary \ref{C-1}.

2). Rewrite Equation \eqref{f-3} in the form
\[
\mathcal{D}_\eta v_\xi = \frac{1}{\xi-\eta}(u_\xi(w_\eta-\eta u_\eta) - (w_\xi-\xi u_\xi)u_\eta).
\]
Dividing both sides of this equation by $\xi^{g-1}$ and taking the limit at $\xi\to\infty$, we obtain
$\mathcal{D}_\eta v_1 = w_\eta-\eta u_\eta -(w_1-u_1)u_\eta$. Therefore, $w_\xi = (\xi+w_1-u_1)u_\xi + \mathcal{D}_\xi v_1$.
Thus, statement (b) also follows from Corollary \ref{C-1}.
\end{proof}

\begin{cor}\text{ }

\textup{(a)}\; $v_k = -\dfrac{1}{2}u_k'$, $k=1,\dots,g$,

\textup{(b)}\; $v_k = -\dfrac{1}{2}\partial_{k-1}u_1$, $k=2,\dots,g$,

\textup{(c)}\; $w_k = u_k+(h_1-2u_1)u_{k-1}-\dfrac{1}{2}u_{k-1}''$, $k=2,\dots,g$,

\textup{(d)}\; $w_{g+1} = (h_1-2u_1)u_g-\dfrac{1}{2}u_g''$, $g\geqslant 1$.
\end{cor}

\begin{cor} \label{C-2}
\text{ }

\textup{(a)}\;  $4u_2 = u_1''+6u_1^2-4h_1u_1+2h_2$,

\textup{(b)}\;  $4w_2 = -u_1''-2u_1^2+2h_2$,

\textup{(c)}\; $u_2' = \partial_2 u_1$,

\textup{(d)}\; $u_2'' = 2(h_1-2u_1)u_2+2(u_1w_2+u_2w_1)-h_3$.
\end{cor}

\begin{cor} \label{C-3}
Let $k\geqslant 3$. Then

\textup{(a)}\; $2u_k= \dfrac{1}{2}u_{k-1}''-\sum_{i+j=k}u_iw_j-\dfrac{1}{4}\sum_{i+j=k-1}u_i''u_j''-(h_1-2u_1)u_{k-1}+h_k$,

\textup{(b)}\; $2w_k= -\dfrac{1}{2}u_{k-1}''-\sum_{i+j=k}u_iw_j-\dfrac{1}{4}\sum_{i+j=k-1}u_i''u_j''+(h_1-2u_1)u_{k-1}+h_k$,

\textup{(c)}\; $u_k'= \partial_ku_1$.
\end{cor}

\begin{cor} \label{C-4}
\[
\frac{1}{2}u_g'' = (h_1-2u_1)u_g+\sum\limits_{i+j=g+1}u_iw_j+\frac{1}{4}\sum\limits_{i+j=g}u_i''u_j''-h_{g+1}.
\]
\end{cor}

%%%%%%%%%%%%%%   5
\section{The $(P,Q)$-recursion}
Note that the statements of Corollaries \ref{C-2} and \ref{C-3} do not depend on a specific $g$. Using this fact, we obtain the key result:
\begin{thm}
For each infinitely differentiable function $u_1$ and any vector $\mathbf{h}=(h_1,h_2,\dots)$, there exists a unique infinite sequence
of differential polynomials
\[
P_k = P_k(u_1,u_1',\ldots,u_1^{(2k-2)};h),\quad Q_k = Q_k(u_1,u_1',\ldots,u_1^{(2k-2)};h),\; k=1,\ldots,
\]
такая, что
\begin{alignat*}2
P_1&= u_1,&\qquad Q_1&= h_1-u_1,\\
P_2&= \frac{1}{4}(u_1''+6u_1^2-4h_1u_1+2h_2),&\qquad Q_2&= \frac{1}{4}(-u_1''-2u_1^2+2h_2),
\end{alignat*}
The polynomials $P_k$ and $Q_k$ for $k\geqslant 3$ are determined  recursively by
\begin{align}
\label{f-23} P_k =& \frac{1}{4}P_{k-1}'' - \frac{1}{2}\sum\limits_{i+j=k}P_iQ_j - \frac{1}{8}\sum\limits_{i+j=k-1}P_i''P_j'' - \frac{1}{2}(h_1-2u_1)P_{k-1} + \frac{1}{2}h_k;  \\
\label{f-24} Q_k =& -\frac{1}{4}P_{k-1}'' - \frac{1}{2}\sum\limits_{i+j=k}P_iQ_j - \frac{1}{8}\sum\limits_{i+j=k-1}P_i''P_j'' + \frac{1}{2}(h_1-2u_1)P_{k-1} + \frac{1}{2}h_k.
\end{align}
\end{thm}

Set $P_k = P_k^0 + P_k^h$, where $P_k^0 = P_k(u_1,u_1',\ldots,u_1^{(2k-2)};0)$.

\begin{exs}\label{P-3-2}
In the case $\mathbf{h}=0$, setting $u_1=u$, we obtain:
\begin{align*}
P_1^0 = &\;  u;\\[5pt]
P_2^0 = &\;  \frac{3u^2}{2} + \frac{u''}{4};\\[5pt]
P_3^0 = &\;  \frac{5u^3}{2} + \frac{5uu''}{4} + \frac{5(u')^2}{8} + \frac{u^{(4)}}{16};\\[5pt]
P_4^0 = &\;  \frac{35u^4}{8} + \frac{35u^2u''}{8} + \frac{35u(u')^{2}}{8} + \frac{7uu^{(4)}}{16} + \frac{7u'u'''}{8} + \frac{21(u'')^2}{32} + \frac{u^{(6)}}{64};\\[7pt]
P_5^0 = &\;  \frac{63u^5}{8} + \frac{105u^3u''}{8} + \frac{315u^2(u')^2}{16} + \frac{63u^2u^{(4)}}{32} + \frac{63uu'u'''}{8} + \frac{189u(u'')^2}{32} + \frac{9uu^{(6)}}{64} + \\[7pt]
 &\; \qquad \qquad \qquad \qquad \qquad \qquad + \frac{231(u')^2u''}{32} + \frac{27u'u^{(5)}}{64} + \frac{57u''u^{(4)}}{64} + \frac{69(u''')^2}{128} + \frac{u^{(8)}}{256}.
\end{align*}

In the case of non-zero $\mathbf{h}$, setting $u_1=u$, we obtain:
\begin{align*}
P_1 = &\; P_1^0 = u; \\
P_2 = &\; P_2^0 - h_1u + \frac{h_2}{2}; \\
P_3 = &\; P_3^0 + h_{1}^2u - \frac{h_1 h_2}{2} - 3h_1u^2 - \frac{h_1u''}{2} + \frac{h_2u}{2} + \frac{h_3}{2}; \\
P_4 = &\; P_4^0 - h_1^3u + \frac{h_1^2 h_2}{2} + \frac{9h_1^2u^2}{2} + \frac{3h_1^2u''}{4} - h_1h_2u - \frac{h_1 h_3}{2} - \frac{15h_1u^3}{2} - \frac{15h_1uu''}{4} - \\
      &\; \qquad \qquad \qquad \qquad \qquad - \frac{15h_1(u')^2}{8} - \frac{3h_1u^{(4)}}{16} - \frac{h_2^2}{8} + \frac{3h_2u^2}{4} + \frac{h_2u''}{8} + \frac{h_3u}{2} + \frac{h_4}{2}.
\end{align*}
\end{exs}

\begin{defn}
We will refer to the recursion that generates the sequence of pairs of differential polynomials $(P_k,Q_k)$, $k=1,2,\dots$, in the function $P_1=u$
according to Equations  \eqref{f-23}--\eqref{f-24} as the $(P,Q)$-\emph{recursion}.
\end{defn}

Directly from the construction of the $(P,Q)$-recursion and the results of \cite{B-23}, one has
\begin{thm}\label{T-3-4}
The $(P,Q)$ recursion provides an infinite system of compatible differential equations $\partial_k P_1 = \partial_1 P_k$, $k=2,3,\dots$,
where $\partial_k$, $k=1,2,\ldots$, which is a sequence of commuting differentiations of the ring $\mathbb{C}[u_1,u_1',\dots; h_1,h_2,\dots]$,
with $\partial_k h_i=0$ for all $k$ and $i$, and $\partial_1=\partial$.
\end{thm}

Using this $(P,Q)$-recursion, we obtain the following result.
\begin{thm}[key result]
Let $(\mathbf{u},\mathbf{v},\mathbf{w})$ be a solution of the system of Equations \eqref{f-2}--\eqref{f-4} and let $P_1,P_2,\ldots$ be
the sequence of functions generated by the $(P,Q)$-recursion with the condition $P_{g+1}=0$. Then the function $u_1 = u_1(t_1,\ldots,t_g)$ is a solution to the hierarchy
\begin{equation*}%\label{f-22-1}
\partial_k u_1 = \partial_1 P_k,\; k = 2,\ldots,g,\; \text{ where } \partial_k u_1 = \frac{\partial u_1}{\partial t_k}.
\end{equation*}
The function $u_1 = u_1(t_1,\dots,t_g)$, as a function of $t_1$, is a solution to the nonlinear ordinary differential equation $P_{g+1} = 0$ of the form
\begin{equation}\label{f-23-1}
 P_g'' = 2(h_1-2u_1)P_g + 2\sum\limits_{i+j=g+1}P_iQ_j + \frac{1}{2}\sum\limits_{i+j=g}P_i''P_j'' - 2h_{g+1}.
\end{equation}
\end{thm}

From the system of Equations \eqref{f-23}--\eqref{f-24}, we have $P_k \approx\frac{1}{4^k}P_1^{(2k-2)}\approx -Q_k$,
where <<$\approx$>> denotes equality in terms of the absolute values of terms containing $P_1^{(n)}$ for all $0\leqslant n<2k-2$.

\begin{cor}
Equation \eqref{f-23-1} is an ordinary differential equation of the function $P_1(t_1)$ of order $2g$ with $g+1$ free parameters $h_1,\dots,h_{g+1}$.
 After multiplying by $4^9$, the highest derivative $P_1^{(29)}$ in this equation has coefficient $1$.
\end{cor}

Thus, according to classical theory of differential equations (see, for example, \cite{Tsic-62}), obtained equation
uniquely determines the function $P_1(t_1)$ in a neighborhood of any point $t_1^*$ where $2g$ initial conditions $P_1(t_1^*),\dots,P_1^{(2g-1)}(t_1^*)$ are specified.

\begin{ex}
For $g=1$, we obtain the equation
\[
u_1'' + 6u_1^2 - 4h_1u_1 + 2h_2 = 0,
\]
the solution of which, under the initial conditions $a = u_1(t_1^*)$ and $b = u_1'(t_1^*)$, uniquely determines the point $(u_1,v_1,w_1,w_2)$
in the phase space of the Mamford dynamical system by the formulas
\[
v_1 = -\frac{1}{2}u_1',\quad w_1 = h_1-u_1,\quad w_2 = h_2-u_1(h_1-u_1).
\]
The initial conditions $u_1(t_1^*),\; u_1'(t_1^*)$ and the parameters $h_1$, $h_2$ determine the value of the integral $h_3 = u_1(t_1^*)w_2(t_1^*) + v_1^2(t_1^*)$.
\end{ex}

Let us introduce a grading on the parameters $h_1,h_2,\dots$ by setting $\deg h_n = 2n$.
Note that in the case of non-zero values of the parameters $h_n$ the differential polynomials from example \ref{P-3-2} are written in the form
\begin{align*}
P_1 = &\, P_1^0 = u; \\
P_2 = &\, P_2^0 - h_1P_1^0 + \frac{1}{2}h_2P_0^0; \\
P_3 = &\, P_3^0 -2h_1P_2^0 + (h_1^2+\frac{h_2}{2})P_1^0 + \frac{1}{2}(h_3-h_1 h_2)P_0^0; \\
P_4 = &\, P_4^0\! -\! 3h_1P_3^0 +(3h_1^2+\frac{1}{2}h_2)P_2^0\! -\! (h_1^3+h_1h_2\!-\!\frac{1}{2}h_3)P_1^0 + \frac{1}{2}(h_1^2h_2\!-\!h_1h_3\!-\!\frac{1}{4}h_2^2+h_4)P_0^0,
\end{align*}
where the coefficients of $P_{k-i}^0$ are homogeneous polynomials in $h_n$ and $P_0^0=1$.

\begin{thm}[P. Baron, see \cite{B-23}]\label{T-B1}
Let $P_k$, $k=1,2,\dots$, be the polynomials obtained by the $(P,Q)$-recursion, and let $P_0=1$. Then for any $k$ we have
\begin{equation}\label{f-P}
P_{k+1} = P_{k+1}^0 + \sum_{i=0}^{k}\alpha_{k+1,i}(\mathbf{h})P_{k-i}^0,
\end{equation}
where $\alpha_{k+1,i}(\mathbf{h})$ are homogeneous polynomials in $h_n$ with coefficients in $\mathbb{Z}[\frac{1}{2}]$.
\end{thm}

%%%%%%%%%%%%%%   6
\section{The Gelfand-Dikii recursion}\text{}

Consider the graded differential ring of polynomials
$$
\mathcal{A} = \bigoplus_{k\geqslant 0}\mathcal{A}_k = \mathbb{C}[u,u',\dots,u^{(k)},\dots]
$$
with the differentiation operator $\partial$, such that $\deg u = 2$, $\deg \partial = 1$, $u' = \partial u$, $u^{(k+1)} = \partial u^{(k)}$, $\deg u^{(k)} = k+2$.

Set $\deg z = 1$ and let $\mathcal{A}[[z^{-1}]]$ be the graded differential ring of formal Laurent series,  i.e. series of form
$\sum_{k \geqslant q} f_k z^{-k}$ with condition $\partial z = 0$.

In the paper \cite{G-D-79} it was proved that the series
\begin{equation}\label{eq-0}
R = z^{-1}\sum\limits_{k \geqslant 0} R_{2k} z^{-2k},\; \text{ where }\; R_0 = \frac{1}{2}\,,\; R_{2k} \in A,
\end{equation}
is a solution to the equation
\begin{equation}\label{eq-1}
-R''' + 4(u+z^2)R' + 2u'R = 0
\end{equation}
if and only if it is also a solution to the equation
\begin{equation}\label{eq-2}
-2RR'' + (R')^2 + 4(u+z^2)R^2 = c(z),
\end{equation}
where $c(z) = 1 + \sum\limits_{k \geqslant 1} c_{2k} z^{-2k}$, and $c_{2k}$ are some arbitrary constants.

The solution to Equation \eqref{eq-2} with $c(z) \equiv 1$ is called the \emph{standard} solution and it is denoted by $R$.
The general solution to this Equation is denoted by $R_c$. We have $R_c=\alpha(z)R$, where $\alpha(z)^2=c(z)$.

In our approach, the grading plays a crucial role.
We assume that $\deg R_{2k} = 2k$, $\deg c_{2k} = 2k$.
Then $\deg R = -1$ and $\deg c(z) = 0$.
In this grading, Equations \eqref{eq-1} and \eqref{eq-2} become homogeneous, and the coefficients $R_{2k}$ become homogeneous polynomials in the graded variables $u,u',\ldots,u^{(k-2)}$.

From Equations \eqref{eq-1} and \eqref{eq-2} we have
$$
R_{2k+2}' = \frac{1}{4}R_{2k}''' - uR_{2k}' - \frac{1}{2}u'R_{2k},
$$
or, on other words,
\begin{equation}
\label{eq-4}
\partial R_{2k+2} = \Lambda\partial R_{2k},
\end{equation}
where $\Lambda = \frac{1}{4}\partial^2 - u - \frac{1}{2}u'\partial^{-1}$ is the Lenard operator.

\underline{Note}: Equation \eqref{eq-4} is not enough to conclude that $R_{2k+2}$ is a differential polynomial in $u$.

Let us show how we can derive a recursion for $R_{2k}$ with the conditions $R_0 = \frac{1}{2}$ and $R_2 = -\frac{1}{4}u$ from Equation \eqref{eq-2} with the condition $c(z) = 1$.

\begin{lem}
Substituting the expression $R = z^{-1}(1/2+\mathcal{R})$ with $\mathcal{R} = \sum_{k \geqslant 1} R_{2k} z^{-2k}$ into the equation
\[
-2RR'' + (R')^2 + 4(u+z^2)R^2 = 1,
\]
i.e., into Equation \eqref{eq-2} with the condition $c(z) = 1$, gives an explicit recursion for the coefficients of the series \eqref{eq-0} that defines the standard solution of Equation \eqref{eq-2}.
\end{lem}
\begin{proof}
After this substitution, we obtain
\[
-(1+2\mathcal{R})\mathcal{R}'' + (\mathcal{R}')^2 + (u+z^2)(1+2\mathcal{R})^2 = z^2.
\]
Therefore,
\begin{equation}\label{eq-5}
4\sum\limits_{k \geqslant 2} (R_{2k}+uR_{2k-2}) z^{-2k+2} = \mathcal{R}''+2\mathcal{R}\mathcal{R}'' - (\mathcal{R}')^2 -
4(u+z^2)\mathcal{R}^2.
\end{equation}
Thus $4R_{2k}\approx R_{2k-2}'',\; k\geqslant 2$, where $\approx$ means equality up to monomials decomposable in the ring $A$.
Therefore, the equation \eqref{eq-5} gives an explicit recursion for the coefficients of the standard solution to the equation \eqref{eq-2}.
\end{proof}

\begin{exs}\label{P-3-6}
\begin{align*}
2R_0 & = 1, \\[1pt]
4R_2 & = -u, \\[1pt]
16R_4 & = 3u^2-u'', \\[1pt]
64R_6 & = -10u^3+10uu''+5(u')^2-u^{(4)}.
\end{align*}
\end{exs}

\begin{defn}
The recursion defined by Equation \eqref{eq-5} is called the \emph{Gelfand-Dikii recursion} (or the \emph{GD-recursion} for short).
\end{defn}

Let's put $\alpha(z) = 1 + \sum\limits_{k\geqslant 1}\alpha_{2k}z^{-2k}$. Then, according to the paper \cite{G-D-79},
the polynomials $R_{c,2k}$, $k=1,2,\dots$, obtained by the GD-recursion, can be expressed in terms of the polynomials
$R_{2k}$, $k=0,1,2,\dots$ of the standard solution of Equation \ref{eq-2} by the formula
\begin{equation}\label{f-R}
R_{c,2k+2} = R_{2k+2} + \sum_{i=1}^k\alpha_{2i}R_{2k+2-2i}.
\end{equation}
From Equations \eqref{f-P} and \eqref{f-R}, we obtain an important corollary.

\begin{cor}
The coefficients $\alpha_{2i}$, $i=0,1,\dots$, at $R_{2k-2i}$ in the GD-recursion do not depend on $k$,
and the coefficients $\alpha_{k+1,i}$, $i=0,1,\dots$, at $P_{k-1}^0$ in the $(P,Q)$-recursion depend on $k$.
\end{cor}

Comparing the formulas from Example \ref{P-3-2}, where $\mathbf{h}=0$, to the formulas from Example \ref{P-3-6},
we obtain that the polynomials $P_k^0$ and $R_{2k}$ are related by
\[
P_k^0(-\frac{1}{2}u) = 2R_{2k}(u),\; k=1,2,3.
\]
\begin{thm}[P. Барон, see \cite{B-23}]\label{T-B2}
For any $k=1,2,\ldots$
\[
P_k^0(-\frac{1}{2}u) = 2R_{2k}(u).
\]
\end{thm}

%%%%%%%%%%%%%%%   7
\section{Korteweg--de Vries hierarchy}
Consider the graded differential ring of polynomials
$$
\mathcal{A} = \bigoplus_{k\geqslant 0}\mathcal{A}_k = \mathbb{C}[u,u',\dots,u^{(k)},\dots].
$$
Suppose that there exists an infinite sequence of homogeneous commutative differentiations $\partial_1=\partial$, $\partial_{2k-1}$, $k=2,3,\dots$, $|\partial_{2k-1}| = 2k-1$, in the ring $\mathcal{A}$.

\begin{lem}\label{L-4-1}
Let $\{C_{2k}^1\}$ and $\{C_{2k}^2\}$ be two sequences of homogeneous polynomials from the ring $\mathcal{A}$,
where $|C_{2k}^1| = |C_{2k}^2|= 2k$, $k=1,2,\dots$, and $C_2^1 = C_2^2$. If
\[
\partial C_{2k+2}^i = \Lambda\partial C_{2k}^i,\qquad i=1,2,
\]
where $\Lambda = \frac{1}{4}\partial^2-u-\frac{1}{2}u'\partial^{-1}$ is the Lenard operator,
then  $C_{2k}^1 = C_{2k}^2$ for all $k>1$.
\end{lem}
\begin{proof}
The homogeneous differential operator  $\Lambda\partial = \frac{1}{4}\partial^3-u\partial-\frac{1}{2}u'$, $|\Lambda\partial| = 3$,
is well-defined in the ring $\mathcal{A}$.
It can be easily shown that, in the ring $\mathcal{A} = \bigoplus_{k\geqslant 0}\mathcal{A}_k$ the operator $\partial$
defines \emph{monomorphisms} $\partial\colon \mathcal{A}_k \to \mathcal{A}_{k+1}$, $k>0$.
Therefore, the Lenard operator $\Lambda$ is well-defined and unique on $\rm\;{Im}\partial$.
Using the condition $C_2^1 = C_2^2$, where $|C_2^1|=2$, and induction on $k$, we obtain the proof.
\end{proof}

\begin{defn}
The \emph{Korteweg-de Vries hierarchy} is an infinite sequence of compatible equations
\begin{equation*}
\partial_{2k-1}U_2 = \partial U_{2k},\qquad k=1,2,\dots,
\end{equation*}
where $U_{2k}$ are differential polynomials of $U_2$ related by
\begin{equation}\label{F-2}
\partial U_{2k+2} = \Lambda\partial U_{2k},\qquad k=1,2,\dots,
\end{equation}
$\Lambda$ is the Lenard operator, and $\partial_{2k-1}$ are commuting differentiations of the ring $\mathcal{A}$.
\end{defn}

Set $U_2 = 4u$ and $U_0 = -8$. Then Equation \eqref{F-2} holds for $k=0$ as well.

\begin{ex}
For $k=1$ and $U_2 = 4u$, we obtain the classical Korteweg--de Vries equation
\begin{equation*}
4\partial_3 u = u''' - 3(u^2)'.
\end{equation*}
\end{ex}

The general solution of the equation $\partial_3U_2 = \partial U_4$ under the condition $\partial U_4 = \Lambda\partial U_2$ with $U_2=4u$ has the form
\begin{equation*}
U_4 = u'' - 3u^2 + \alpha_1u + \alpha_2,
\end{equation*}
where $\alpha_1, \alpha_2 \in \mathbb{C}$, meaning it is not homogeneous. A homogeneous solution would be
\[
U_4 = u'' - 3u^2.
\]

\begin{defn}
A solution $\{U_{2k}^0,\,k\geqslant 1\}$ of the KdV hierarchy in which all polynomials $U_{2k}^0\in \mathcal{A}$ are homogeneous,
$|U_{2k}^0| = 2k$, and $U_2=4u$, is called \emph{special}.
\end{defn}

From Theorems \ref{T-3-4} and \ref{T-B2}, we obtain

\begin{thm}\label{T-spets}
The sequence of homogeneous polynomials $\{P_k^0,\; k\geqslant 1\}$ defined by the $(P,Q)$-recursion with the zero
vector of parameters $\mathbf{h}$ gives a special solution $\{\mathcal{U}_{2k}^0\}$ of the KdV hierarchy, where
\begin{equation}\label{F-5}
\mathcal{U}_{2k}^0 = -8P_k^0(-u/2),\qquad k=1,2,\dots\,.
\end{equation}
\end{thm}

From Lemma \ref{L-4-1}, we obtain
\begin{cor}
The special solution of the KdV hierarchy is uniquely defined and determined by polynomials of the form \eqref{F-5}.
\end{cor}

\begin{cor}
Let $\{C_{2k},\; k=1,2,\dots\}$ be a sequence of homogeneous polynomials in $\mathcal{A}$, where $|C_{2k}|= 2k$, and let
$\{\partial^C_{2k-1}$, $k=1,2,\dots\}$ be a sequence of homogeneous differentiations of the ring $\mathcal{A}$,
each of which commutes with the operator $\partial_1^C$. Then, if $C_2=4u$ and
\[
\partial_{2k-1}^C C_2 = \partial C_{2k},\quad \partial C_{2k+2} = \Lambda\partial C_{2k},\qquad k=1,2,\dots,
\]
then the operators $\partial_{2k-1}^C$ pairwise commute, and $C_{2k} = \mathcal{U}_{2k}^0$.
\end{cor}
\begin{proof}
The equation $\partial_{2k-1}^C C_2 = \partial C_{2k}$ and the condition $C_2=4u$ uniquely determine homogeneous differentiations $\partial_{2k-1}^C$
of the ring $\mathcal{A}$, with condition that they commute with the operator $\partial$.
Since $\partial C_{2k+2} = \Lambda\partial C_{2k}$, then according to Lemma \ref{L-4-1}, we obtain $C_{2k} = \mathcal{U}_{2k}^0$
and $\partial_{2k-1}^C = \partial_{2k-1}$. Therefore, the operators $\partial_{2k-1}^C$ commute for all $k$.\qed
\end{proof}

Using Theorem \ref{T-B2}, we odtain
\begin{cor}
The standard solution $\{R_{2k}\}$ of the GD-recursion defines a special solution of the KdV hierarchy
by the formula $\mathcal{U}_{2k}^0 = -16R_{2k}$, $k=1,2,\dots$.
\end{cor}

This Corollary provides a new proof of a well-known result by Gelfand--Dikii, see \cite{G-D-79}.

Consider an infinite sequence $\{\alpha_{2k}$, $k=1,2,\dots\}$ of graded parameters with $|\alpha_{2k}|= 2k$, and set $\alpha_0=1$.
Denote by $\mathcal{A}^\alpha$ the graded ring $\mathcal{A}\otimes \mathbb{C}[\alpha_2,\alpha_4,\dots] = \bigoplus_{k\geqslant 0}\mathcal{A}_k^\alpha$,
where $\mathcal{A}_k^\alpha$ is a finite-dimensional space with a basis of monomials with grading $k$.

An infinite family of homogeneous commutative differentiations $\partial_{2k-1}$ of the ring $\mathcal{A}$ defines an infinite family of
homogeneous commutative differentiations $\{D_{2k-1}^\alpha\}$ of the ring $\mathcal{A}^\alpha$ where $D_1 = \partial$, and
\[
D_{2k-1}^\alpha = \partial_{2k-1}^0 + \sum_{i=1}^k \alpha_{2i}\partial_{2k-2i+1}^0,\qquad k>1.
\]

For example, $D_3^\alpha = \partial_3^0+\alpha_2\partial_1^0$ and $D_5^\alpha = \partial_5^0+\alpha_2\partial_3^0+\alpha_4\partial_1^0$.

Let's introduce an infinite family of homogeneous polynomials $\{\mathcal{U}_{2k}^\alpha\}$ in $\mathcal{A}^\alpha$, where $|\mathcal{U}_{2k}^\alpha|= 2k$,
\[
\mathcal{U}_{2k}^\alpha = \mathcal{U}_{2k}^0 + \sum_{i=1}^{2k} \alpha_{2i}\mathcal{U}_{2k-2i}^0,
\]
$\mathcal{U}_0^0 = -8$, and $\{\mathcal{U}_{2k}^0 \}$, $k=1,2,\dots$, is a sequence that defines a special solution
of the KdV hierarchy with respect to the operators $\partial_{2k-1}$, $k=1,2,\dots$.

For example, $\mathcal{U}_2^\alpha = \mathcal{U}_2^0+\alpha_2 = 4u+\alpha_2$ and $\mathcal{U}_4^\alpha = \mathcal{U}_4^0+\alpha_2\mathcal{U}_2^0+\alpha_4$.

\begin{thm}\label{T-7-9}
Let $\{U_{2k}\}$, $k=1,2,\dots$, be a general solution of the Korteweg--de Vries hierarchy where $U_2=4u$ and $U_0=-8$.
Then there exists the unique defined ring homomorphism $\pi \colon \mathbb{C}[\alpha_2,\alpha_4,\dots]\to \mathbb{C}$
that induces a ring homomorphism $\pi \colon \mathcal{A}^\alpha \to \mathcal{A}$ such that
\[
\pi_*D_{2k-1}^\alpha = \partial_{2k-1}\quad\text{and}\quad\pi_*\mathcal{U}_{2k}^\alpha = U_{2k}.
\]
\end{thm}
\textit{Proof} directly follows from the previous results.

\begin{cor}
The sequence of homogeneous polynomials $\{P_{k},k\geqslant 1\}$ defined by the $(P,Q)$-recursion with $\mathbf{h}=(h_1,h_2,\ldots)$,
where $h_i$ are algebraically independent parameters, gives a general parametric solution $\{\mathcal{U}_{2k}^\alpha\}$, $k=1,2,\dots$
of the Korteweg--de Vries hierarchy by the formula
\[
\mathcal{U}_{2k}^\alpha = -8P_k\left(-\frac{1}{2}u\right).
\]
\end{cor}

%%%%%%%%%%%%%%%   8
\section{Hyperelliptic curves}
Consider a $2g$-dimensional family of curves
\[
 V_{\lambda} = \{ (X,Y) \in \mathbb{C}^2\colon Y^2 = F(X)\},
\]
where $F(X) = 4X^{2g+1} + \lambda_4 X^{2g-1} + \cdots + \lambda_{4g-2}$.

Let $\mathcal{D} = \{ \lambda\in\mathbb{C}^{2g}:F(X)\,\text{ has multiple roots}\}$, and set $\mathcal{B} = \mathbb{C}^{2g}\setminus \mathcal{D}$.
For each $\lambda\in\mathcal{B}$, we obtain a smooth hyperelliptic curve $\bar{V}_{\lambda}$ of genus $g$
with Jacobian $\rm{Jac}(\bar{V}_{\lambda}) = \mathbb{C}^g\!/\Gamma_g$.
Here, $\Gamma_g\subset\mathbb{C}^g$ is a lattice of rank $2g$ generated by periods of the integrals of $g$ holomorphic differentials
over $2g$ cycles, which form a basis for the one-dimensional homologies of the curve $\bar{V}_{\lambda}$.

In 1886, Klein posed the problem of constructing hyperelliptic functions of genus $g>1$, that possess properties similar to properties of Weierstrass elliptic functions.
To solve this problem, it was required to construct a hyperelliptic sigma function $\sigma(\mathbf{z};\lambda)$, which is an analog of the Weierstrass elliptic sigma function.
In 1898, H.F. Baker addressed Klein's problem for $g=2$, see \cite{Baker-1898}, see also \cite{BEL-19}.
The case $g>2$ remained a problem for a long time, attention to which intensified under the influence
of the development of algebraic-geometric methods in the theory of solitons.

In the paper \cite{BEL-97-2}, it was shown that there exists a unique defined function $\sigma(\mathbf{z};\lambda)$
in $\mathbf{z} = (z_1,\ldots,z_{2g-1})\in \mathbb{C}^g$, where $\lambda=(\lambda_4,\ldots,\lambda_{4g+2})\in \mathbb{C}^{2g}$,
which is called the hyperelliptic sigma function.
In the neighborhood of the point $\mathbf{z}=0$, the coefficients of the series expansion of the function $\sigma(\mathbf{z};\lambda)$
with respect to $\mathbf{z}$ are polynomials in $\lambda$.
The logarithmic derivatives of this function of order $2$ and higher generate the field of meromorphic functions on the Jacobian $\rm{Jac}(\bar{V}_{\lambda})$.
Such, this function $\sigma(\mathbf{z};\lambda)$ is a solution to Klein's problem, see \cite{BEL-19}.

Systems of multidimensional heat equations in a nonholonomic frame (systems of multidimensional Schrodinger equations) were constructed in \cite{BL-04}.
It was shown that these systems uniquely determine sigma functions for a wide class of so-called $(n,s)$-curves.
In the papers \cite{B-Bun-20} and \cite{B-Bun-23} (вased on papers \cite{BShor-04} and \cite{BL-04}) it were presented explicit formulas for the operators
$Q_0, Q_2,\ldots, Q_{4g-2}$, $g\geqslant 1$, that define the systems determining sigma functions of hyperelliptic curves, i.e., $(2,2g+1)$-curves.
It should be noted that, according to \cite{B-Bun-20}, the operators $Q_0, Q_2,\ldots, Q_{4g-2}$ generate a polynomial Lie algebra
that has only 3 generators $Q_0, Q_2, Q_4$ for any $g>1$.

Set $f'(\mathbf{z}) = \frac{\partial}{\partial z_1}f(\mathbf{z})$ and
\[
\wp_{2k} = -\frac{\partial^2}{\partial z_1\partial z_{2k-1}}\ln\sigma, k=1,\ldots,g;\quad
\wp_{2i-1,2k-1} = -\frac{\partial^2}{\partial z_{2i-1}\partial z_{2k-1}}\ln\sigma, i\neq 1,\, k\neq 1.
\]
\begin{thm}[см. \cite{BEL-97-2}]
All algebraic relations among the derivatives of the function $\ln\sigma(z)$ with respect to $z_1,\ldots,z_{2g-1}$ of order $2$ and higher
are generated by the relations
\begin{equation} \label{f-25}
\wp''_{2i} = 6(\wp_{2i+2} +\wp_{2}\wp_{2i}) - 2(\wp_{3,2i-1} - \lambda_{2i+2} \delta_{i,1}).
\end{equation}
\vspace{-10mm}
\begin{multline}\label{f-26}
\wp'_{2i}\wp'_{2k} = 4(\wp_{2i}\wp_{2k+2} + \wp_{2i+2}\wp_{2k} + \wp_{2}\wp_{2i}\wp_{2k} + \wp_{2i+1,2k+1}) - \\[3pt]
- 2(\wp_{2i}\wp_{3,2k-1} + \wp_{2k}\wp_{3,2i-1} + \wp_{2i-1,2k+3} +\wp_{2i+3,2k-1}) + \\[3pt]
+ 2(\lambda_{2i+2}\wp_{2k}\delta_{i,1} + \lambda_{2k+2}\wp_{2i}\delta_{k,1}) + 2\lambda_{2(i+j+1)} (2\delta_{i,k} + \delta_{i,k-1} + \delta_{i-1,k}).
\end{multline}
Here $\delta_{i,k}$ is the Kronecker symbol.
\end{thm}

\begin{cor}
For all $g\geqslant 1$ we have
\begin{enumerate}
  \item [(a)] When $i=1$, Equation \eqref{f-25} implies $\wp''_{2} = 6\wp_{2}^2 + 4\wp_{4} + 2\lambda_4$.\\[5pt]
  \item [(b)] When $i=2$, Equation \eqref{f-25} implies $\wp''_{4} = 6(\wp_{2} \wp_{4} + \wp_{6}) - 2\wp_{3,3}$.\\[5pt]
  \item [(c)] When $i=k=1$, Equation \eqref{f-26} implies $(\wp'_{2})^2 = 4[\wp_{2}^3+(\wp_{4}+\lambda_4)\wp_{2}+\wp_{3,3}-\wp_{6}+\lambda_6]$.
\end{enumerate}
\end{cor}

We have $\wp_{2i}' = \partial_{2i-1}\wp_{2}$. Then from (a) we obtain
\begin{cor}
For any $g>1$ the function $u = 2\wp_2(\mathbf{t})$ is a solution to the Korteweg--de Vries equation
\[
u''' = 6 u u' + 4\dot u,\; \text{ где }\; \dot u = 2\frac{\partial u}{\partial z_3}.
\]
\end{cor}	

Set
\[
x_\xi = \xi^g - \sum\limits_{i=1}^g \wp_{2i}\xi^{g-i},\qquad  y_\xi = \sum\limits_{i=1}^g \wp_{2i}'\xi^{g-i},\qquad  z_\xi = \sum\limits_{i=1}^g \wp_{2i}''\xi^{g-i}.
\]
\begin{thm}
Hyperelliptic Klein functions provide a solution to the equation
\[
\mathcal{D}_\eta L_\xi = \frac{1}{\xi-\eta}[L_\xi,L_\eta] + [L_\xi,A_\eta],
\]
where $L_\xi =
\begin{pmatrix}
  v_\xi & u_\xi \\
  w_\xi & -v_\xi
\end{pmatrix}$,\;
$A_\eta =
\begin{pmatrix}
  0 & 0 \\
  u_\eta & 0
\end{pmatrix}$ and $u_\xi = 2x_\xi$,\; $v_\xi = y_\xi$,\; $w_\xi = z_\xi+2x_\xi(\xi+2\wp_{2})$.
\end{thm}
Note that $u_\xi = 2(\xi^g-\wp_{2}\xi^{g-1}+\ldots)$\; and $w_\xi =
2\xi^{g+1}+2\wp_{2}\xi^{g}+(\wp_{2}''-2\wp_4-4\wp_{2}^2)\xi^{g-1}+\ldots$

%%%%%%%%%%%%%%

\section*{Acknowledgments}
I would like to thank P.G. Grinevich, A.V. Domrin, A.V. Mikhailov, S.P. Novikov, V.N. Rubtsov, V.V. Sokolov, and A.V. Tsyganov
for valuable discussions on the results of this work.
Special thanks to Polina Baron for assistance with examples and for Theorems \ref{T-B1} and \ref{T-B2},
the proofs of which are provided in her paper \cite{B-23} in this issue of the journal.

%\newpage

\end{document}